\documentclass[runningheads]{llncs}
\usepackage{graphicx}
\usepackage{makecell}
\usepackage{amsmath}
\usepackage{amssymb}
\usepackage{caption}
\usepackage{hyperref}
\usepackage{amsfonts}
\usepackage{adjustbox}
\usepackage{url}
\usepackage{tabularx}
\usepackage{algorithm}
\usepackage{fontawesome5}
\usepackage[noend]{algpseudocode}
\usepackage[caption=false]{subfig}
\usepackage{pgfplots}
\DeclareMathOperator{\E}{\mathbb{E}}
\begin{document}
\title{Representing Higher-Order Networks\\with Spectral Moments}
\author{Hao Tian\inst{1}\faIcon{envelope}\and
Shengmin Jin\inst{2}\thanks{The work was done prior to the author joining Amazon} \and
Reza Zafarani\inst{1}}
\authorrunning{H.Tian et al.}
\institute{Data Lab, EECS Department, Syracuse University, Syracuse, NY 13244, USA\\\email{\{haotian,reza\}@data.syr.edu} \and
Amazon.com Inc., New York, USA\\
\email{\{shengmij\}@amazon.com}}
\maketitle              % typeset the header of the contribution
\begin{abstract}
% Limitation 200 words
   The spectral properties of traditional (\textit{dyadic}) graphs, where an edge connects exactly two vertices, are widely studied in different applications. These spectral properties are closely connected to the structural properties of dyadic graphs. We generalize such connections and characterize higher-order networks by their spectral information. We first split the higher-order graphs by their ``edge orders" into several uniform hypergraphs. For each uniform hypergraph, we extract the corresponding spectral information from the transition matrices of carefully designed random walks. From each spectrum, we compute the first few spectral moments and use all such spectral moments across different ``edge orders" as the higher-order graph representation. We show that these moments not only clearly indicate the return probabilities of random walks but are also closely related to various higher-order network properties such as degree distribution and clustering coefficient. Extensive experiments show the utility of this new representation in various settings. For instance, graph classification on higher-order graphs shows that this representation significantly outperforms other techniques.

\end{abstract}
\section{Introduction}
%background of HON
Graphs are natural representations of interactions among entities. With the demand to model more complex data, the topological structures of graphs are also enriched. For example, in basic undirected graphs, directions or edge weights can be further added to ``featurize" the interactions, leading to \textit{directed} or \textit{weighted} graphs. More recently, types of vertices and edges have been further expanded. Furthermore, many interactions go beyond pairwise in the real world, such as coauthorships or group activities. To model these interactions, the graphs are extended to \textit{Higher-Order Networks (HON)} (see a survey in \cite{10.1145/3682112.3682114}) in which each edge captures a \underline{set} of interacting nodes. Such representations allow edges to have a flexible \textit{order/size} (number of nodes in the edge) in a network. Existing mathematical topologies such as \textit{simplicial complexes} and \textit{hypergraphs} can be used directly for modeling higher-order networks. 

%difficulty of study HON
Although ways to model higher-order interactions are natural, exploring and mining higher-order networks is much more difficult than traditional \textit{dyadic} graphs, which have edges between pairs of nodes. Directly exploring higher-order graphs has many obstacles, including computational complexity and structural sparsity. In an effort to address these difficulties, current research on higher-order networks can be roughly divided into two directions. One is to study higher-order patterns in dyadic graphs, denoted as \textit{network motifs}~\cite{Milo824}. Given a specific motif (small subgraph), higher-order patterns can be learned by counting the motif in the graph~\cite{leskovec2010signed} or identifying the nodes that are part of the 
motif~\cite{DBLP:journals/corr/BensonGL16a}. However, one cannot distinguish actual higher-order interactions from a combination of pairwise interactions in a dyadic graph. Another branch of studies directly models higher-order networks as hypergraphs, but these studies are often limited to, or \textit{downgrade} to, some fixed upper bound on the size/order of the edge. 

%Our approach is..
 \vspace{2mm}\noindent \textbf{The present work: Spectral Representation of Higher-Order Networks.} Our goal is to characterize higher-order networks by extracting and studying their spectral information. In dyadic graphs, the spectral moments of the random walk transition matrix can provide informative machine learning features~\cite{inproceedingsJin}. These spectral moments not only clearly capture the return probabilities of corresponding walks but are also closely related to various network properties such as degree distribution and clustering coefficient. In higher-order networks, random walks can be defined in multiple ways. To obtain a computationally feasible transition matrix, we rely on a strict definition of an \textit{$s$-walk} in higher-order graphs~\cite{lu2011highordered}. More specifically, within a given order $r$ (edges consisting of $r$ nodes), an $s$-walk ($s<r$) consists of a series of order-$r$ edges where any two consecutive edges in the walk overlap on $s$ nodes. The attractive part of this definition is that \textit{there exists an equivalent traditional walk for any $s$-walk in some corresponding weighted dyadic graphs}. By considering the random walk transition matrices for all possible $r$ and $s$ and computing their spectral moments, we obtain a single feature vector to represent the entire higher-order graph. Note that there exist a few studies of node representation learnings for higher-order networks such as~\cite{10.1145/3184558.3186900} (see Section \ref{sec:related}); however, these studies are completely different from ours, as they focus on the similarities between vertices in these graphs. 

%the main contribution
In summary, the main contributions of this work are as follows:\vspace{-2mm}
\begin{itemize}
    \item To the best of our knowledge, we present the first work to represent higher-order networks in their entirety with their spectral information;
    \item We prove how spectral moments capture properties of the original higher-order graph, e.g., degrees and clustering properties; and
    \item Through extensive experiments on higher-order graph classification, we verify the effectiveness of such representations. We further assess the importance of the number of spectral moments utilized and the graph sizes in the proposed representation. 
\end{itemize}

%\noindent \textbf{Organization:} The remainder of the paper is organized as follows. We detail the related work in Section~\ref{sec:related} from the perspective of spectral graph theory in higher-order graphs and representation learnings in dyadic graphs. In Section~\ref{sec:perliminaries}, we review the notation and introduce most related preliminaries. In Section~\ref{sec:theory}, we define the spectral moments in higher-order graphs. In Section~\ref{sec:preprocessing}, we introduce the steps to subgraph sampling and moments extraction, followed by the graph classification experiments on both dyadic and higher-order graphs in Section~\ref{sec:experiments}. We conclude with a discussion and potential future work in Section \ref{sec:discussion}. 

\section{Related Work}
\label{sec:related}
\textbf{Higher-Order Spectral Graph Theory.} 
Laplacian for $r$-uniform hypergraphs was first defined by Chung~\cite{chung1993laplacian} in 1993. In that setting, two edges are `adjacent' when they share exactly $(r-1)$ nodes, a property that can be represented by a new `adjacency' matrix. This definition was extended by Lu and Peng~\cite{lu2011highordered} to work for any number of shared nodes other than $(r-1)$, where each $r$-uniform hypergraph can have $(r-1)$ Laplacians to account for $1$-walks to $(r-1)$-walks. 

\vspace{2mm}\noindent\textbf{Dyadic Graph Representation Learning.}
There exist many techniques to measure the similarity between two dyadic graphs. One main branch is \textit{kernel methods}~\cite{DBLP:journals/corr/abs-1903-11835}, which can compute pairwise similarities between graphs, and the corresponding kernel matrix can be used as a representation for graph classification. Another branch is representation learning methods that extract feature vectors from subgraph patterns~\cite{DBLP:journals/corr/NarayananCVCLJ17} or some target-guided attributes~\cite{pmlr-v48-niepert16}. However, all such methods generate relative matrices/embeddings, having the following drawbacks: (1) the derived single feature vector is not explainable, and (2) the graph representation cannot be easily extended to larger graphs. 

There also exist some adhoc representation learning methods. To represent a graph with features, spectral properties are widely used as they are order-invariant and are closely related to the structural properties, such as the Laplacian spectrum~\cite{DBLP:journals/corr/abs-1912-00735}, spectral moments~\cite{jin2022interpretable,inproceedingsJin} and spectral paths (variations of moments from sampled subgraphs)~\cite{10.1145/3534678.3539433}. Calculating spectral features for a single graph is computationally efficient and hence, we explore the possibility of extending such studies to higher-order graphs in this work.\vspace{-3mm}

\section{Preliminaries}
\label{sec:perliminaries}
\textbf{Notations.}
We briefly summarize some notations used throughout this paper. 
\begin{center}
\begin{adjustbox}{width=0.8\textwidth}
\begin{tabular}{ |ll| } 
 \hline
 $r$ -- order of uniform hypergraphs & $m$ -- spectral moments \\
 $s$ -- order of random walks & $l$ -- order of spectral moments\\
 $\mu$ -- eigenvalues of normalized Laplacian matrix &$k$ -- lengths of vectors or walks\\
 $\lambda$ -- eigenvalues of random walk transition matrix&$G$ -- dyadic graph\\
 $\mathcal{G}$ -- higher-order graph&\\
 \hline
\end{tabular}
\end{adjustbox}
\end{center}

\vspace{2mm}\noindent \textbf{Hypergraph.}
We use $\mathcal{G} = (V, E)$ to denote a hypergraph, where the vertex set $V = \{v_1, v_2, \dots,v_n\}$ is the set of vertices. The edge set $E=\{e|e\subseteq V \}$ is the set of subsets of $V$. We only consider undirected hypergraphs. 

An \textit{$r$-uniform} hypergraph is one where each edge contains $r$ vertices. 

\vspace{2mm} \noindent \textbf{Higher-Order Network.} Slightly different from the hypergraph, we define the higher-order network as follows:
\begin{align*}
	\mathcal{G} = (V, E_{1}, E_{2}, \dots),
\end{align*}
where $V$ is the set of vertices, and  $E_i$,
$$E_{i}\subseteq\{(x_1,\dots,x_i) \;|\; (x_1,\dots,x_i) \in V^i \text{ and } x_1 \neq \dots \neq x_i\}$$ 
is the set of order-$i$ edges. Each \textit{layer} $(V,E_r)$ is an \textit{$r$-uniform} hypergraph, which is a subgraph of $\mathcal{G}$.

\vspace{1mm} \noindent \textbf{Higher-Order Graph Representation Learning.} The representation learning task aims to generate $k$ dimensional vectors from the whole graph, i.e., learn a mapping: $\mathcal{G} \rightarrow \mathbb{R}^k$. Such representations should preserve graph properties and can be utilized in downstream machine learning tasks.  \vspace{-2mm}

%% Above end at page 3

\section{Spectral Moments as Graph Representations}
\label{sec:theory}
We first briefly review the spectral properties of dyadic graphs and analyze how spectral moments are related to their properties. Next, we extend these results to spectral moments of higher-order networks by studying Laplacians of equivalent dyadic graphs. Finally, we prove that these moments provide upper bounds for certain higher-order graph properties. \vspace{-2mm}

\subsection{Spectral Moments of Dyadic Graphs}
\vspace{-2mm}\textbf{Nomalized Laplacian Matrix.} For dyadic graph $G$, let $A$ denote its adjacency matrix and $D=\text{diag}(d_1,d_2,...,d_n)$ its degree matrix. The \textit{Normailized Laplacian} is defined as $L=I-D^{-\frac{1}{2}}AD^{-\frac{1}{2}}$. The eigenvalues $\mu_i$'s of $L$ are bounded: $0 = \mu_1 \leq \mu_2 \leq ... \leq \mu_{n} \leq 2$. 

\vspace{2mm} \noindent \textbf{Random Walk Transition Matrix.} The \textit{transition matrix} of the random walk process on $G$ is $P=D^{-1}A$. As $P$ is \textit{similar} to $I-L$, which is $D^{-\frac{1}{2}}AD^{-\frac{1}{2}}$, the spectrum of $P$ is also bounded $1 = \lambda_1 \geq \lambda_2 \geq ... \geq \lambda_n \geq -1$, as $\lambda_i = 1-\mu_i$ ($1\leq i\leq n$). 

\vspace{2mm} \noindent \textbf{Spectral Moments.} 
Based on the random walk transition matrix, the spectral moments $m_l~(l=1,2,...)$ are defined as~\cite{inproceedingsJin}: \vspace{-2mm}
\begin{align}
	m_l = \mathbb{E}(\lambda^l) = \frac{1}{n}\sum_{i=1}^{n} {\lambda_i}^l.
	\label{eq2:spectral_moments}
\end{align}
Essentially, the $l$-th spectral moment is the expected return probability of an $l$-step random walk starting from a randomly chosen node. These moments are also closely related to various graph properties. For example, the second moment $m_2$ can be bounded by the mean and variance of the degree distribution~\cite{inproceedingsJin}, and it can be used to measure network robustness~\cite{jin2022spectral}. \vspace{-2mm}

\subsection{Spectral Moments of Higher-Order Graphs}
\label{sec4.2:spectral_moments_HON}
To extract spectral moments from higher-order graphs, we transform the higher-order graph into multiple weighted dyadic graphs so that we preserve structural properties at all orders.  

\vspace{2mm} \noindent \textbf{$S$-Walk.} Let $(V, E_r) \subseteq \mathcal{G}$ be an $r-$uniform subgraph of hypergraph $\mathcal{G}$.  For $1 \leq s \leq r-1$, an \textit{s-walk}~\cite{lu2011highordered} of length $k$ is a sequence of vertices:
$$v_1, v_2,..., v_j ,..., v_{(r-s)(k-1)+r}$$
along with a sequence of edges $e_1, e_2,..., e_k$ such that \vspace{-2mm}
$$e_i = \{v_{(r-s)(i-1)+1},v_{(r-s)(i-1)+2},...,v_{(r-s)(i-1)+r}\}.$$
Simply speaking, an $s$-walk of $r$-uniform hypergraph is a sequence of $r$-edges where each two consecutive edges in the walk share $s$ number of vertices. 

\begin{figure}[t]
    \subfloat[A $2$-walk in an $r=4$ hypergraph.\label{fig:r4s2_walk}]{%
			\resizebox{0.5\linewidth}{!}{\begin{tikzpicture}

%%% boundaries
\draw[ultra thin] (0,0) -- (7.8,0) -- (7.8,2) -- (0,2) -- (0,0);
\draw[ultra thin] (4,-0.2) -- (12,-0.2) -- (12,2.2) -- (4,2.2) -- (4,-0.2);
\draw[ultra thin] (8.2,0) -- (16,0) -- (16,2) -- (8.2,2) -- (8.2,0);

%%% nodes
\node[font=\fontsize{52}{58}\sffamily\bfseries] at (1,0.8) {a};
\node[font=\fontsize{52}{58}\sffamily\bfseries] at (3,1) {b};
\node[font=\fontsize{52}{58}\sffamily\bfseries] at (5,0.8) {c};
\node[font=\fontsize{52}{58}\sffamily\bfseries] at (7,1) {d};
\node[font=\fontsize{52}{58}\sffamily\bfseries] at (9,0.8) {e};
\node[font=\fontsize{52}{58}\sffamily\bfseries] at (11,1) {f};
\node[font=\fontsize{52}{58}\sffamily\bfseries] at (13,0.9) {g};
\node[font=\fontsize{52}{58}\sffamily\bfseries] at (15,1) {h};
\node[font=\fontsize{52}{58}\sffamily\bfseries] at (19,1) {...    ...};

\end{tikzpicture}}%
		}\hfill
    \subfloat[A $2$-walk in an $r=3$ hypergraph.\label{fig:r3s2_walk}]{%
			\resizebox{0.43\linewidth}{!}{\begin{tikzpicture}

%%% boundaries
\draw[ultra thin] (0,0.2) -- (5.8,0.2) -- (5.8,1.8) -- (0,1.8) -- (0,0.2);
\draw[ultra thin] (2,0) -- (8,0) -- (8,2) -- (2,2) -- (2,0);
\draw[ultra thin] (4,-0.2) -- (10,-0.2) -- (10,2.2) -- (4,2.2) -- (4,-0.2);
\draw[ultra thin] (6.2,0.2) -- (12,0.2) -- (12,1.8) -- (6.2,1.8) -- (6.2,0.2);

%%% nodes
\node[font=\fontsize{52}{58}\sffamily\bfseries] at (1,0.8) {a};
\node[font=\fontsize{52}{58}\sffamily\bfseries] at (3,1) {b};
\node[font=\fontsize{52}{58}\sffamily\bfseries] at (5,0.8) {c};
\node[font=\fontsize{52}{58}\sffamily\bfseries] at (7,1) {d};
\node[font=\fontsize{52}{58}\sffamily\bfseries] at (9,0.8) {e};
\node[font=\fontsize{52}{58}\sffamily\bfseries] at (11,1) {f};
\node[font=\fontsize{52}{58}\sffamily\bfseries] at (15,1) {...    ...};

\end{tikzpicture}}%
		}
	
	\caption{Two cases of $s$-walks when calculating Laplacians. (a) Case $1 \leq s \leq r/2$, edges are only adjacent to one previous and one following edge; (b) Case $r/2 < s \leq r-1$, there exists more overlaps among edges.
    \label{fig:swalks}}

\end{figure}
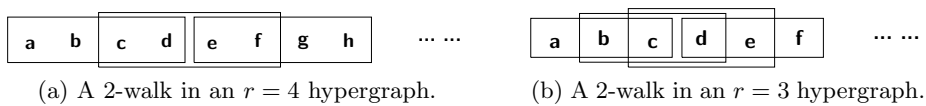

\vspace{2mm} \noindent \textbf{$S$-Laplacian.} Normalized Laplacian of $r-$uniform hypergraphs can be obtained by constructing the corresponding dyadic graphs. Basically, we convert the $s$-walks on higher-order edges to traversals on dyadic graphs with the same path lengths. To build a one-to-one mapping from higher-order to dyadic, the new nodes in dyadic graphs have to be ordered tuples of size $s$. More specifically, we define an undirected (weighted) graph $G^{(s)}$ based on $V^{\underline{s}}$, where $V^{\underline{s}}$ is the set of all ordered $s$-tuples of $V$. Note that all $s$-tuples have to be in order to maintain the order of the $s$-walks. 

In the following two cases, as shown in Figure~\ref{fig:swalks}, the converted graphs $G^{(s)}$ will be slightly different with respect to $s$. Figure~\ref{fig:swalks} (a) leads to dyadic path \texttt{(ab)-(cd)-(ef)-(gh)}; while Figure~\ref{fig:swalks} (b) leads to \texttt{(ab)-(bc)-(cd)-(de)-(ef)}. Here, we define the Laplacian of the generated graph $G^{(s)}$ as the \textit{$S$-Laplacian} of $r$-uniform hypergraph $(V, E_r)$.

\vspace{2mm} \noindent \textbf{Degrees in hypergraph and dyadic graph.} When $1 \leq s \leq r/2$, for $x \in V^{\underline{s}}$, the degree of $x$ in dyadic graph $G^{(s)}$ is given by $d_x = \sum_y w(x,y) = D_{[x]} {{r-s}\choose s} s!,$ where $D_{[x]}$ is the degree of the set $[x]$ in hypergraph $\mathcal{G}$, which is the number of hyperedges containing set $[x]$. $d_x$ is equal to $D_{[x]}$ multiplied by permutations of the remaining $(r-s)$ nodes.  \vspace{-2mm}

\subsection{Spectral Moments and Properties of Higher-Order Network}
\label{subsec:proofs}
\vspace{-2mm}Here, we aim to show that spectral moments can capture the structural information of higher-order networks. First, we extend past results on spectral moments in simple undirected graphs to weighted graphs, which facilitates studying the $s$-walk transformed graphs from the hypergraph. Next, we discuss how the second and third spectral moments provide bounds on the node-set degrees and clustering properties in the hypergraph. Here, we only prove cases under $1 \leq s \leq r/2$, as when $s > r/2$, discussing graph properties becomes challenging due to the repetition of some of the original nodes.\vspace{-3mm}

\subsubsection{Spectral moments in weighted graphs}
Here, we take the second moment $m_2$ as an example. In \cite{inproceedingsJin}, the authors prove that the second spectral moment $m_2$ of an undirected and unweighted graph is $m_2 = \E(d_i)\E(\frac{1}{d_id_j})$. In Theorem \ref{Thm:Weighted}, we show that the equation also applies to weighted graphs.\vspace{-2mm}
% its random walk transition matrix eq Theorem \ref{Thm:Weighted}, we first prove the $m_2$ value of weighted graphs.

\begin{theorem}
\label{Thm:Weighted}
For a weighted graph, the second spectral moment $m_2$ is:\vspace{-2mm}
$$m_2 = \E(\lambda^2) = \E(d_i)\E(\frac{1}{d_id_j}),$$ 
where $\E(d_i)$ denotes the average degree in the weighted graph and $d_id_j$ follows the joint degree distribution $p(d_i,d_j)$: the probability that a node with degree $d_i$ is connected to another node with degree $d_j$.
\end{theorem}

\begin{proof}
For any edge $(i,j)$ with weight $k$, we can consider it as $k$ unit edges with weight 1 between nodes $i$ and $j$. As mentioned, the $2^\text{nd}$ spectral moment can be viewed as the expected return probability of a 2-step random walk, which is equal to the summation of the return probability of a 2-step random walk starting from any node $i$ divided by the number of nodes. For an edge connecting nodes $i$ and $j$, it will increase the overall return probability by $\frac{2}{d_id_j}$ as it includes 2 closed walks: $i \rightarrow j \rightarrow i$ and $j \rightarrow i \rightarrow j$. There are $\frac{n}{2} \cdot \E(d_i) \cdot p(d_i,d_j)$ edges with node having degreee $d_i$ and $d_j$. Therefore, 
\begin{align}
\label{Equation1}
    m_2 = \frac{\sum_{d_i, d_j} \frac{2}{d_id_j} \cdot \frac{n}{2} \cdot \E(d_i) \cdot p(d_i,d_j)}{n} = \E(d_i)\E(\frac{1}{d_id_j}).
\end{align}
\end{proof}

\subsubsection{Second spectral moment $m_2$ and degrees.} In Section~\ref{sec4.2:spectral_moments_HON}, we converted the higher-order graph to multiple weighted dyadic graphs. However, the degrees in converted graphs are not exactly the degrees in the original hypergraph. Here, we demonstrate the relationship between $m_2$ and degrees of the $r$-uniform hypergraph.

\begin{theorem}
\label{Thm:HOm2}
    For an $r$-uniform hypergraph $\mathcal{G} = (V, E_r)$, when $1 \leq s \leq r/2$, the second spectral moment $m_2$ of the corresponding dyadic graph $G = (V_G, E_G)$ is equal to $\frac{{r \choose s}|E_r|}{2{r-s \choose s}|V_G|}  \cdot \E_{(i,j) \in G}(\frac{1}{D{[i]}D{[j]}})$, where $D_{[i]}$ and $D_{[j]}$ are the actual degree of node sets $[i]$ and $[j]$ in the hypergraph, respectively.
\end{theorem}

\begin{proof}
For a node $i$ in the dyadic graph $G$, we denote $d_i$ as its degree in dyadic graph $G$, and $D_{[i]} = |{e \in E_r}: i \subset e|$ as its degree in the hypergraph $\mathcal{G}$. In other words, $D_{[i]}$ is the number of actual hyperedges where node set $[i]$ is in.

From Equation \ref{Equation1}, we can know that $m_2 = \frac{\sum_{i \in V_G}d_i}{|V_G|} \cdot \sum_{i \sim j}\frac{1}{d_id_j} \cdot p(d_i,d_j)$. Notice that for each hyperedge in $\mathcal{G}$ when $1 \leq s \leq r/2$, it can generate ${r \choose s} s! \cdot {r-s \choose s} s!/ 2$ edges in the corresponding dyadic graph. Therefore, $\sum_{i \in V_G}d_i = {r \choose s} s! \cdot {r-s \choose s} s! \cdot |E_r|$. Moreover, for each node $i$, $d_i = D_{[i]} {r-s \choose s} s!$, so $\frac{1}{d_id_j} = \frac{1}{{r-s \choose s}^2 s!^2} \cdot \frac{1}{D_{[i]}D_{[j]}}$ and $p(d_i, d_j) = p(D_{[i]}, D_{[j]})$. Hence, 

\begin{align}
    m_2 &= \frac{{r \choose s}s!{r-s \choose s}s! |E_r|} {2|V_G|} \cdot \sum_{i \sim j} \frac{1}{{r-s \choose s}^2 s!^2} \cdot \frac{1}{D_{[i]}D_{[j]}} \cdot p(D_{[i]}, D_{[j]}) \\
    &= \frac{{r \choose s}|E_r|}{2{r-s \choose s}|V_G|}  \cdot \E_{(i,j) \in G}(\frac{1}{D_{[i]}D_{[j]}})
\end{align}
\end{proof}

Theorem \ref{Thm:HOm2}, provides a natural bound: 
\begin{equation}
\label{ineq:m2}
    m_2 \geq \frac{1}{2{r-s \choose s}} \cdot \E_{(i,j) \in G}(\frac{1}{D_{[i]}D_{[j]}}),
\end{equation}
which indicates that $m_2$ provides an upper bound on $\E_{(i,j) \in G}(\frac{1}{D_{[i]}D_{[j]}})$.
% explain why r choose s Er > VG

\subsubsection{Third spectral moment $m_3$ and clustering properties.}
Similar to Theorem~\ref{Thm:Weighted}, it is easy to prove in a weighted graph $m_3 = 2\E(\delta_i)\E(\frac{1}{d_hd_id_j})$~\cite{inproceedingsJin}, where $\E(\delta_i)$ is the average number of triads a node is in and $d_hd_id_j$
follows the joint degree distribution of triads $p(d_h, d_i
, d_j)$ (we skip the details for brevity). Since $G$ becomes a weighted graph, when node $h,i,j$ form a triangle $\delta$, the weight of $\delta$ is counted by $\min(w(h,i),w(h,j),w(j,i))$. 
Here, we justify how $m_3$ also represents the clustering and degree properties of the original hypergraph. 

\begin{theorem}
\label{Thm:HOm3}
    For an $r$-uniform hypergraph $\mathcal{G} = (V, E_r)$, when $1 \leq s \leq r/2$, the third spectral moment $m_3$ of the corresponding dyadic graph $G = (V_G, E_G)$ is equal to $2(\E(\Delta_{[i]})+3\cdot{r-s \choose 2}\cdot \frac{|E_r|}{|V|}) \cdot \frac{1}{{r-s \choose s}^3} \cdot \E_{(h,i,j) \in G}(\frac{1}{D_{[h]}D_{[i]}D_{[j]}})$, where $\E(\Delta_{[i]})$ is the average number of triads of node-set $[i]$ belong to in the hypergraph, $\E_{(h,i,j) \in G}(\frac{1}{D_{[h]}D_{[i]}D_{[j]}})$ is the joint degree distribution of node set $[h], [i] ,[j]$, $|V|$ and $|E_r|$ are the number of nodes and edges of the original hypergraph. 
\end{theorem}

\begin{proof}
For a node set $i$ in the hypergraph $\mathcal{G}$, we denote $\Delta_{[i]}$ as the number of triads it belongs to. We denote two triads as different when at least one of their hyperedges is different. First, we calculate how many triads are formed in the transformed graph $G$. If $e_1,e_2,e_3$ are three hyperedges that form a triangle restricted by walk of size $s$, and $e_1\cap e_2=i$, $e_2\cap e_3=j$ and $e_3\cap e_1=h$. The number of triads in the dyadic graph is multiplied by only the permutations of $i$, $j$ and $h$, that is 
\begin{equation}
    |\delta_i| =s!^3(|\Delta_{[i]}|+3\cdot{r-s \choose 2}\cdot \frac{|E_r|}{|V|})
\end{equation}

Then, $m_3$ can be calculated by 
\begin{align}
    m_3 &= 2\E(\delta_i)\E(\frac{1}{d_hd_id_j}) \\
     &= 2s!^3(|\Delta_{[i]}|+3\cdot{r-s \choose 2}\cdot \frac{|E_r|}{|V|}) \cdot \frac{1}{{r-s \choose s}^3 s!^3} \cdot \E_{(h,i,j) \in G}(\frac{1}{D_{[h]}D_{[i]}D_{[j]}})\\
     &= 2(|\Delta_{[i]}|+3\cdot{r-s \choose 2}\cdot \frac{|E_r|}{|V|}) \cdot \frac{1}{{r-s \choose s}^3} \cdot \E_{(h,i,j) \in G}(\frac{1}{D_{[h]}D_{[i]}D_{[j]}})
\end{align}
\end{proof}

Theorem \ref{Thm:HOm3} provides two inequalities (they sum up to  $m_3$):
\begin{align}
 m_3 & \geq \frac{2|\Delta_{[i]}|}{{r-s \choose s}^3}\cdot \E_{(h,i,j) \in G}(\frac{1}{D_{[h]}D_{[i]}D_{[j]}})\label{eq:4.3.1}\\
    m_3 & \geq 6\cdot{r-s \choose 2}\cdot \frac{|E_r|}{|V|} \cdot \frac{1}{{r-s \choose s}^3} \cdot \E_{(h,i,j) \in G}(\frac{1}{D_{[h]}D_{[i]}D_{[j]}})\label{eq:4.3.2}
\end{align}
Inequality \ref{eq:4.3.1} provides bound on number of triangles and degree distribution, while inequality \ref{eq:4.3.2} provides bound on number of nodes, number of edges and degrees distribution. Generally, for well-connected real-world graphs, inequality \ref{eq:4.3.1} will be a closer bound, as validated in section \ref{cha:quality_bound}. 
\begin{table*}[t]
\centering
\begin{minipage}{1\linewidth} % <--- 
\centering
\begin{adjustbox}{width=0.9\textwidth}
\begin{tabular}{@{}|l|r|r|r|r|r|r|@{}}
\hline
Graphs & Vertices & Timestamps & Unique Edges & Max Order & Average Order & Class\\
\hline
coauth-DBLP~\cite{DBLP:journals/corr/abs-1802-06916} & 1,924,991 & 3,700,067 & 2,599,087 & 25 & 2.78 & coauth\\
coauth-MAG-Geology~\cite{DBLP:journals/corr/abs-1802-06916} & 1,256,385 & 1,590,335 & 1,207,390 & 25 & 2.78 & coauth\\
coauth-MAG-History~\cite{DBLP:journals/corr/abs-1802-06916} & 1,014,734 & 1,812,511 & 895,668 & 25 & 1.31 & coauth\\
congress-bills~\cite{DBLP:journals/corr/abs-1802-06916} & 1,718 & 260,851 & 85,082 & 25 & 3.66 & congress\\
contact-high-school~\cite{DBLP:journals/corr/abs-1802-06916} & 327 & 172,035 & 7,937 & 5 & 2.05 & contact\\
contact-primary-school~\cite{DBLP:journals/corr/abs-1802-06916} & 242 & 106,879 & 12,799 & 5 & 2.10 & contact\\
DAWN~\cite{DBLP:journals/corr/abs-1802-06916} & 2,558 & 2,272,433 & 143,523 & 16 & 1.58 & dawn\\
email-Enron~\cite{DBLP:journals/corr/abs-1802-06916} & 143 & 10,883 & 1,542 & 18 & 2.47 & email\\
email-Eu~\cite{DBLP:journals/corr/abs-1802-06916} & 998 & 234,760 & 25,791 & 25 & 2.33 & email\\
tags-ask-ubuntu~\cite{DBLP:journals/corr/abs-1802-06916} & 3,029 & 271,233 & 151,441 & 5 & 2.71 & tags\\
tags-math-sx~\cite{DBLP:journals/corr/abs-1802-06916} & 1,629 & 822,059 & 174,933 & 5 & 2.19 & tags\\
tags-stack-overflow~\cite{DBLP:journals/corr/abs-1802-06916} & 49,998 & 14,458,875 & 5,675,497 & 5 & 2.97& tags\\
threads-ask-ubuntu~\cite{DBLP:journals/corr/abs-1802-06916} & 125,602 & 192,947 & 167,001 & 14 & 1.80& threads\\
twitter-hashtag-covid19~\footnote{\small \url{www.trackmyhashtag.com/data/COVID-19.zip}} & 12,033 & 59,892 & 10,074 & 33 & 2.21& twitter\\
twitter-hashtag-ira~\footnote{\small \url{https://archive.org/details/twitter-ira}} & 190,481 & 2,585,982 & 242,988 & 30 & 1.44& twitter\\
\hline
\end{tabular}
\end{adjustbox}
\end{minipage}
\caption{Data Statistics\vspace{-8mm}}
\label{tab1:summary_of_dataset}
\end{table*}

\vspace{-3mm}\section{Experiments}
\label{sec:experiments}
\vspace{-3mm}We validate the usefulness of spectral moments through extensive experiments, verifying the bounds on the second and third spectral moments ($m_2$ and $m_3$) and assessing their utility for graph classification. For classification, we simply feed the extracted features to the XGBoost~\cite{Chen:2016:XST:2939672.2939785} classifier with the default parameters. We further assess the influence of the number of moments and subgraph sizes. 

\vspace{-4mm}\subsection{Experimental Setup}
\label{sec:secExperimental}
\vspace{-2mm}We detail our experimental setup in this section. Our preprocessing includes (1) sampling many subgraphs from the higher-order graph; (2) computing spectral moments for each subgraph; and (3) \textit{downgrading} the higher-order graph to a dyadic graph for baseline methods.

\vspace{2mm} \noindent \textbf{Dataset and Subgraph Sampling.}
We use 15 publicly available higher-order graphs under 8 labels, as shown in Table~\ref{tab1:summary_of_dataset}. We only consider unique edges (as opposed to, for example, using weights/timestamps). Orders indicate the number of vertices that an edge can have (i.e., edge size). The class is the category of the dataset, where a graph with the same class is expected to have similar structures.  

\begin{figure*}[t]

    \subfloat[2nd moment $m_2$ bounds degrees.\label{fig:m2_bound}]{%
			\resizebox{0.45\linewidth}{4cm}{% This file was created with tikzplotlib v0.10.1.
\begin{tikzpicture}

\definecolor{darkgray176}{RGB}{176,176,176}
\definecolor{lightgray204}{RGB}{204,204,204}

\begin{axis}[
height=8cm,
legend cell align={left},
legend style={
  fill opacity=0.8,
  draw opacity=1,
  text opacity=1,
  at={(0.03,0.97)},
  anchor=north west,
  draw=lightgray204
},
tick align=outside,
tick pos=left,
width=10cm,
x grid style={darkgray176},
xlabel={Graphs},
xmin=-1.89, xmax=30.89,
xtick style={color=black},
y grid style={darkgray176},
ymin=0, ymax=0.525,
ytick style={color=black}
]
\draw[draw=none,fill=blue] (axis cs:-0.4,0) rectangle (axis cs:0.4,0.414285714285714);
\addlegendimage{ybar,ybar legend,draw=none,fill=blue}
\addlegendentry{$m_2$}

\draw[draw=none,fill=blue] (axis cs:0.6,0) rectangle (axis cs:1.4,0.428743961352657);
\draw[draw=none,fill=blue] (axis cs:1.6,0) rectangle (axis cs:2.4,0.472972972972973);
\draw[draw=none,fill=blue] (axis cs:2.6,0) rectangle (axis cs:3.4,0.428);
\draw[draw=none,fill=blue] (axis cs:3.6,0) rectangle (axis cs:4.4,0.395408163265306);
\draw[draw=none,fill=blue] (axis cs:4.6,0) rectangle (axis cs:5.4,0.486842105263158);
\draw[draw=none,fill=blue] (axis cs:5.6,0) rectangle (axis cs:6.4,0.41036036036036);
\draw[draw=none,fill=blue] (axis cs:6.6,0) rectangle (axis cs:7.4,0.443356643356643);
\draw[draw=none,fill=blue] (axis cs:7.6,0) rectangle (axis cs:8.4,0.446666666666667);
\draw[draw=none,fill=blue] (axis cs:8.6,0) rectangle (axis cs:9.4,0.481481481481481);
\draw[draw=none,fill=blue] (axis cs:9.6,0) rectangle (axis cs:10.4,0.395438596491228);
\draw[draw=none,fill=blue] (axis cs:10.6,0) rectangle (axis cs:11.4,0.425170068027211);
\draw[draw=none,fill=blue] (axis cs:11.6,0) rectangle (axis cs:12.4,0.5);
\draw[draw=none,fill=blue] (axis cs:12.6,0) rectangle (axis cs:13.4,0.482758620689655);
\draw[draw=none,fill=blue] (axis cs:13.6,0) rectangle (axis cs:14.4,0.5);
\draw[draw=none,fill=blue] (axis cs:14.6,0) rectangle (axis cs:15.4,0.43482905982906);
\draw[draw=none,fill=blue] (axis cs:15.6,0) rectangle (axis cs:16.4,0.398700305810398);
\draw[draw=none,fill=blue] (axis cs:16.6,0) rectangle (axis cs:17.4,0.423958333333333);
\draw[draw=none,fill=blue] (axis cs:17.6,0) rectangle (axis cs:18.4,0.460227272727273);
\draw[draw=none,fill=blue] (axis cs:18.6,0) rectangle (axis cs:19.4,0.415064102564103);
\draw[draw=none,fill=blue] (axis cs:19.6,0) rectangle (axis cs:20.4,0.415849673202614);
\draw[draw=none,fill=blue] (axis cs:20.6,0) rectangle (axis cs:21.4,0.474226804123711);
\draw[draw=none,fill=blue] (axis cs:21.6,0) rectangle (axis cs:22.4,0.46523178807947);
\draw[draw=none,fill=blue] (axis cs:22.6,0) rectangle (axis cs:23.4,0.428579418344519);
\draw[draw=none,fill=blue] (axis cs:23.6,0) rectangle (axis cs:24.4,0.372474747474748);
\draw[draw=none,fill=blue] (axis cs:24.6,0) rectangle (axis cs:25.4,0.451149425287356);
\draw[draw=none,fill=blue] (axis cs:25.6,0) rectangle (axis cs:26.4,0.495049504950495);
\draw[draw=none,fill=blue] (axis cs:26.6,0) rectangle (axis cs:27.4,0.458333333333333);
\draw[draw=none,fill=blue] (axis cs:27.6,0) rectangle (axis cs:28.4,0.45414201183432);
\draw[draw=none,fill=blue] (axis cs:28.6,0) rectangle (axis cs:29.4,0.46656050955414);
\draw[draw=none,fill=red,fill opacity=0.5] (axis cs:-0.2,0) rectangle (axis cs:0.2,0.298432208514755);
\addlegendimage{ybar,ybar legend,draw=none,fill=red,fill opacity=0.5}
\addlegendentry{bounded value}

\draw[draw=none,fill=red,fill opacity=0.5] (axis cs:0.8,0) rectangle (axis cs:1.2,0.340207389132553);
\draw[draw=none,fill=red,fill opacity=0.5] (axis cs:1.8,0) rectangle (axis cs:2.2,0.40405836753306);
\draw[draw=none,fill=red,fill opacity=0.5] (axis cs:2.8,0) rectangle (axis cs:3.2,0.348152356902357);
\draw[draw=none,fill=red,fill opacity=0.5] (axis cs:3.8,0) rectangle (axis cs:4.2,0.284583836553945);
\draw[draw=none,fill=red,fill opacity=0.5] (axis cs:4.8,0) rectangle (axis cs:5.2,0.443397745571658);
\draw[draw=none,fill=red,fill opacity=0.5] (axis cs:5.8,0) rectangle (axis cs:6.2,0.335354069752663);
\draw[draw=none,fill=red,fill opacity=0.5] (axis cs:6.8,0) rectangle (axis cs:7.2,0.371060533277638);
\draw[draw=none,fill=red,fill opacity=0.5] (axis cs:7.8,0) rectangle (axis cs:8.2,0.374466777075178);
\draw[draw=none,fill=red,fill opacity=0.5] (axis cs:8.8,0) rectangle (axis cs:9.2,0.432599487785658);
\draw[draw=none,fill=red,fill opacity=0.5] (axis cs:9.8,0) rectangle (axis cs:10.2,0.287037037037036);
\draw[draw=none,fill=red,fill opacity=0.5] (axis cs:10.8,0) rectangle (axis cs:11.2,0.306702256141433);
\draw[draw=none,fill=red,fill opacity=0.5] (axis cs:11.8,0) rectangle (axis cs:12.2,0.5);
\draw[draw=none,fill=red,fill opacity=0.5] (axis cs:12.8,0) rectangle (axis cs:13.2,0.402050953483245);
\draw[draw=none,fill=red,fill opacity=0.5] (axis cs:13.8,0) rectangle (axis cs:14.2,0.5);
\draw[draw=none,fill=red,fill opacity=0.5] (axis cs:14.8,0) rectangle (axis cs:15.2,0.36479303781133);
\draw[draw=none,fill=red,fill opacity=0.5] (axis cs:15.8,0) rectangle (axis cs:16.2,0.343502821583179);
\draw[draw=none,fill=red,fill opacity=0.5] (axis cs:16.8,0) rectangle (axis cs:17.2,0.290531015037593);
\draw[draw=none,fill=red,fill opacity=0.5] (axis cs:17.8,0) rectangle (axis cs:18.2,0.396864935587761);
\draw[draw=none,fill=red,fill opacity=0.5] (axis cs:18.8,0) rectangle (axis cs:19.2,0.31409918316564);
\draw[draw=none,fill=red,fill opacity=0.5] (axis cs:19.8,0) rectangle (axis cs:20.2,0.308137905223095);
\draw[draw=none,fill=red,fill opacity=0.5] (axis cs:20.8,0) rectangle (axis cs:21.2,0.400098998369439);
\draw[draw=none,fill=red,fill opacity=0.5] (axis cs:21.8,0) rectangle (axis cs:22.2,0.375053057759493);
\draw[draw=none,fill=red,fill opacity=0.5] (axis cs:22.8,0) rectangle (axis cs:23.2,0.354274660524661);
\draw[draw=none,fill=red,fill opacity=0.5] (axis cs:23.8,0) rectangle (axis cs:24.2,0.25024641577061);
\draw[draw=none,fill=red,fill opacity=0.5] (axis cs:24.8,0) rectangle (axis cs:25.2,0.337895622895623);
\draw[draw=none,fill=red,fill opacity=0.5] (axis cs:25.8,0) rectangle (axis cs:26.2,0.449225865209472);
\draw[draw=none,fill=red,fill opacity=0.5] (axis cs:26.8,0) rectangle (axis cs:27.2,0.369995501574449);
\draw[draw=none,fill=red,fill opacity=0.5] (axis cs:27.8,0) rectangle (axis cs:28.2,0.345546348272891);
\draw[draw=none,fill=red,fill opacity=0.5] (axis cs:28.8,0) rectangle (axis cs:29.2,0.376247438296355);

\end{axis}
\end{tikzpicture}}%
		}\hfill
    \subfloat[3rd moment $m_3$ bounds triangles.\label{fig:m3_bound}]{%
			\resizebox{0.457\linewidth}{4cm}{% This file was created with tikzplotlib v0.10.1.
\begin{tikzpicture}

\definecolor{darkgray176}{RGB}{176,176,176}
\definecolor{lightgray204}{RGB}{204,204,204}

\begin{axis}[
height=8cm,
legend cell align={left},
legend style={
  fill opacity=0.8,
  draw opacity=1,
  text opacity=1,
  at={(0.03,0.97)},
  anchor=north west,
  draw=lightgray204
},
tick align=outside,
tick pos=left,
width=10cm,
x grid style={darkgray176},
xlabel={Graphs},
xmin=-1.89, xmax=30.89,
xtick style={color=black},
y grid style={darkgray176},
ymin=0, ymax=0.2291796875,
ytick style={color=black},
yticklabel style={/pgf/number format/fixed}
]
\draw[draw=none,fill=blue] (axis cs:-0.4,0) rectangle (axis cs:0.4,0.164634146341463);
\addlegendimage{ybar,ybar legend,draw=none,fill=blue}
\addlegendentry{$m_3$}

\draw[draw=none,fill=blue] (axis cs:0.6,0) rectangle (axis cs:1.4,0.140444862155388);
\draw[draw=none,fill=blue] (axis cs:1.6,0) rectangle (axis cs:2.4,0.108095844918134);
\draw[draw=none,fill=blue] (axis cs:2.6,0) rectangle (axis cs:3.4,0.130208333333334);
\draw[draw=none,fill=blue] (axis cs:3.6,0) rectangle (axis cs:4.4,0.115759408602151);
\draw[draw=none,fill=blue] (axis cs:4.6,0) rectangle (axis cs:5.4,0.143300653594771);
\draw[draw=none,fill=blue] (axis cs:5.6,0) rectangle (axis cs:6.4,0.127697897340755);
\draw[draw=none,fill=blue] (axis cs:6.6,0) rectangle (axis cs:7.4,0.117066798941799);
\draw[draw=none,fill=blue] (axis cs:7.6,0) rectangle (axis cs:8.4,0.155536529680365);
\draw[draw=none,fill=blue] (axis cs:8.6,0) rectangle (axis cs:9.4,0.120213293650794);
\draw[draw=none,fill=blue] (axis cs:9.6,0) rectangle (axis cs:10.4,0.114088050314465);
\draw[draw=none,fill=blue] (axis cs:10.6,0) rectangle (axis cs:11.4,0.125395906819518);
\draw[draw=none,fill=blue] (axis cs:11.6,0) rectangle (axis cs:12.4,0.170362903225807);
\draw[draw=none,fill=blue] (axis cs:12.6,0) rectangle (axis cs:13.4,0.135089285714286);
\draw[draw=none,fill=blue] (axis cs:13.6,0) rectangle (axis cs:14.4,0.104383680555556);
\draw[draw=none,fill=blue] (axis cs:14.6,0) rectangle (axis cs:15.4,0.103009259259259);
\draw[draw=none,fill=blue] (axis cs:15.6,0) rectangle (axis cs:16.4,0.0971248670919727);
\draw[draw=none,fill=blue] (axis cs:16.6,0) rectangle (axis cs:17.4,0.157843137254902);
\draw[draw=none,fill=blue] (axis cs:17.6,0) rectangle (axis cs:18.4,0.125801282051282);
\draw[draw=none,fill=blue] (axis cs:18.6,0) rectangle (axis cs:19.4,0.154389880952381);
\draw[draw=none,fill=blue] (axis cs:19.6,0) rectangle (axis cs:20.4,0.119763513513514);
\draw[draw=none,fill=blue] (axis cs:20.6,0) rectangle (axis cs:21.4,0.127302918118467);
\draw[draw=none,fill=blue] (axis cs:21.6,0) rectangle (axis cs:22.4,0.139038231780167);
\draw[draw=none,fill=blue] (axis cs:22.6,0) rectangle (axis cs:23.4,0.131098484848485);
\draw[draw=none,fill=blue] (axis cs:23.6,0) rectangle (axis cs:24.4,0.109975961538461);
\draw[draw=none,fill=blue] (axis cs:24.6,0) rectangle (axis cs:25.4,0.143229166666666);
\draw[draw=none,fill=blue] (axis cs:25.6,0) rectangle (axis cs:26.4,0.145559210526316);
\draw[draw=none,fill=blue] (axis cs:26.6,0) rectangle (axis cs:27.4,0.19921875);
\draw[draw=none,fill=blue] (axis cs:27.6,0) rectangle (axis cs:28.4,0.133569587628866);
\draw[draw=none,fill=blue] (axis cs:28.6,0) rectangle (axis cs:29.4,0.155222140402553);
\draw[draw=none,fill=red,fill opacity=0.5] (axis cs:-0.2,0) rectangle (axis cs:0.2,0.107306848602023);
\addlegendimage{ybar,ybar legend,draw=none,fill=red,fill opacity=0.5}
\addlegendentry{bounded value}

\draw[draw=none,fill=red,fill opacity=0.5] (axis cs:0.8,0) rectangle (axis cs:1.2,0.118269141076272);
\draw[draw=none,fill=red,fill opacity=0.5] (axis cs:1.8,0) rectangle (axis cs:2.2,0.0878387637930371);
\draw[draw=none,fill=red,fill opacity=0.5] (axis cs:2.8,0) rectangle (axis cs:3.2,0.104373852332955);
\draw[draw=none,fill=red,fill opacity=0.5] (axis cs:3.8,0) rectangle (axis cs:4.2,0.0666959081376759);
\draw[draw=none,fill=red,fill opacity=0.5] (axis cs:4.8,0) rectangle (axis cs:5.2,0.105645095946875);
\draw[draw=none,fill=red,fill opacity=0.5] (axis cs:5.8,0) rectangle (axis cs:6.2,0.0989697037639312);
\draw[draw=none,fill=red,fill opacity=0.5] (axis cs:6.8,0) rectangle (axis cs:7.2,0.0950675168864363);
\draw[draw=none,fill=red,fill opacity=0.5] (axis cs:7.8,0) rectangle (axis cs:8.2,0.126898660120881);
\draw[draw=none,fill=red,fill opacity=0.5] (axis cs:8.8,0) rectangle (axis cs:9.2,0.0893484167498684);
\draw[draw=none,fill=red,fill opacity=0.5] (axis cs:9.8,0) rectangle (axis cs:10.2,0.0862451018873657);
\draw[draw=none,fill=red,fill opacity=0.5] (axis cs:10.8,0) rectangle (axis cs:11.2,0.0997659356771847);
\draw[draw=none,fill=red,fill opacity=0.5] (axis cs:11.8,0) rectangle (axis cs:12.2,0.0918314255983349);
\draw[draw=none,fill=red,fill opacity=0.5] (axis cs:12.8,0) rectangle (axis cs:13.2,0.101847143295725);
\draw[draw=none,fill=red,fill opacity=0.5] (axis cs:13.8,0) rectangle (axis cs:14.2,0.0638672015108799);
\draw[draw=none,fill=red,fill opacity=0.5] (axis cs:14.8,0) rectangle (axis cs:15.2,0.041944697249481);
\draw[draw=none,fill=red,fill opacity=0.5] (axis cs:15.8,0) rectangle (axis cs:16.2,0.0729940825530712);
\draw[draw=none,fill=red,fill opacity=0.5] (axis cs:16.8,0) rectangle (axis cs:17.2,0.10629272297154);
\draw[draw=none,fill=red,fill opacity=0.5] (axis cs:17.8,0) rectangle (axis cs:18.2,0.0517616648175539);
\draw[draw=none,fill=red,fill opacity=0.5] (axis cs:18.8,0) rectangle (axis cs:19.2,0.125479077732535);
\draw[draw=none,fill=red,fill opacity=0.5] (axis cs:19.8,0) rectangle (axis cs:20.2,0.0935417550904313);
\draw[draw=none,fill=red,fill opacity=0.5] (axis cs:20.8,0) rectangle (axis cs:21.2,0.103345717275927);
\draw[draw=none,fill=red,fill opacity=0.5] (axis cs:21.8,0) rectangle (axis cs:22.2,0.114393705632734);
\draw[draw=none,fill=red,fill opacity=0.5] (axis cs:22.8,0) rectangle (axis cs:23.2,0.0959079177567959);
\draw[draw=none,fill=red,fill opacity=0.5] (axis cs:23.8,0) rectangle (axis cs:24.2,0.0438181860207099);
\draw[draw=none,fill=red,fill opacity=0.5] (axis cs:24.8,0) rectangle (axis cs:25.2,0.0885499694247641);
\draw[draw=none,fill=red,fill opacity=0.5] (axis cs:25.8,0) rectangle (axis cs:26.2,0.11627059137999);
\draw[draw=none,fill=red,fill opacity=0.5] (axis cs:26.8,0) rectangle (axis cs:27.2,0.00732421874999994);
\draw[draw=none,fill=red,fill opacity=0.5] (axis cs:27.8,0) rectangle (axis cs:28.2,0.11622041543177);
\draw[draw=none,fill=red,fill opacity=0.5] (axis cs:28.8,0) rectangle (axis cs:29.2,0.131643833532202);
\end{axis}

\end{tikzpicture}}%
		}
	
	\caption{$m_2$ and $m_3$ via certain bouned values calculated from hypergraph properties, such as degrees and number of triangles. Figure \ref{fig:m2_bound} shows inequality of Equation \ref{ineq:m2}, Figure \ref{fig:m3_bound} shows inequality of Equation \ref{eq:4.3.1}. For clearness only 30 out of 500 small graphs from coauth-DBLP are shown. }
    \label{fig:moments_bound}
\end{figure*}
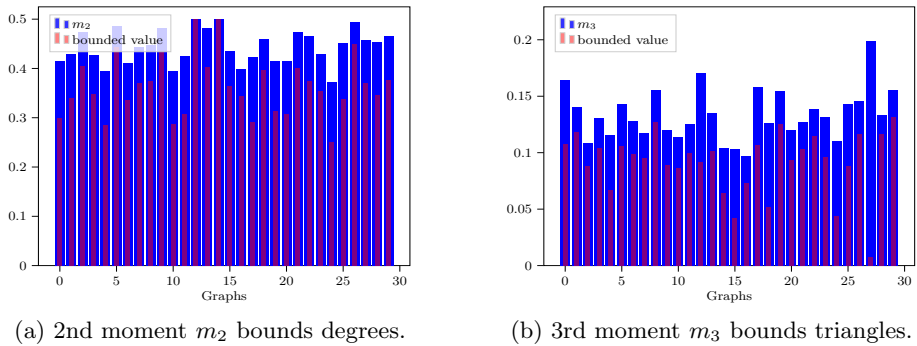

\vspace{-3mm}To perform classification on many (but relatively smaller) higher-order graphs, we perform random walk sampling~\cite{NIPS2006_dff8e9c2} on each graph. Starting from a random node, we select an edge uniformly at random and then move to another node on the \underline{same edge}. After the required sample size is reached, all edges that are composed solely of the sampled nodes will be selected. The advantage of random walk sampling is that it keeps a relatively complete structure locally and maintains high levels of connectivity. In our main experiment, the size of the randomly sampled subgraphs is in the range of 50-200. 

\vspace{1mm}\noindent \textbf{Computing Spectral Moments.}
As discussed in Section~\ref{sec4.2:spectral_moments_HON}, we calculate the moments by building dyadic graphs from the edge set. Specifically, for each higher-order edge, we generate all permutations of the given size $s$, as the new nodes of the dyadic graph. We form new edges if any pair of such nodes meet the requirements of being strictly involved in an $s$-walk. As the minimum of maximum order among all datasets is 5, we switch the order $r$ from 2 to 5 and $1\leq s\leq r-1$ correspondingly.

\vspace{1mm}\noindent \textbf{Downgrading to Dyadic Graphs.}
There is no absolute whole-graph representation learning method for higher-order networks, and extending classic dyadic methods to higher-order ones is challenging. To justify the effectiveness of the proposed method, we \textit{downgrade} higher-order graphs to dyadic graphs and apply multiple state-of-the-art classification methods on these dyadic variants. Specifically, to downgrade, for each higher-order edge, we enumerate all of its dyadic sub-edges and add them to the downgraded graph of the same vertex space. For example, an order-3 edge $(v_1,v_2,v_3)$ will yield three downgraded edges: $(v_1,v_2)$, $(v_2,v_3)$, and $(v_1,v_3)$~\cite{10.1145/3394486.3403060}. While downgrading to a dyadic graph loses information on higher-order connections, it preserves structure in the lower-order space. 

\vspace{-3mm}\subsection{Quality of Bounds in Real-World Graphs.} \label{cha:quality_bound} \vspace{-2mm}Although second $m_2$ and third $m_3$ spectral moments provide bounds by certain graph properties, the quality of these bounds requires further verification through experiments on real-world graphs. Here, we further measure how the actual values are close to these bounds. 

For each individual higher-order graph, we enumerate the actual number of triangles and calculate degrees to compute the actual value bounded by spectral moments -- the right-hand side of equations \ref{ineq:m2} and \ref{eq:4.3.1}. Figure \ref{fig:moments_bound} shows the first 30 graphs from coauth-DBLP, \ref{fig:moments_bound}(a) captures equation \ref{ineq:m2} while \ref{fig:moments_bound}(b) captures equation \ref{eq:4.3.1}. In most cases, the calculated values from graph properties are not only bounded but also close to spectral moments, which indicates the effectiveness of $m_2$ and $m_3$. However, there are still a few cases in which the bounded values are much lower than the spectral moments. \vspace{-4mm}

\subsection{Higher-Order Graph Classification}
\vspace{-2mm}In order to test the effectiveness of spectral moments, we perform various graph classification tasks on the generated subgraphs of 15 higher-order graphs we collected. There are eight types of classes in total, so we select the first graph from each class and put their subgraphs into a dataset. For classes that contain more than one graph, we further perform classifications on their subgraphs to test whether the proposed method can distinguish subgraphs of similar categories. 

Similarly, we perform the same classification on the \textit{downgraded} dyadic graphs (as discussed in Section \ref{sec:secExperimental}), applying the following state-of-the-art methods~\footnote{For the first two kernel methods, we applied GraKeL, and for the last two methods, we directly used the code provided by the authors.}: Shortest Path Kernel~\cite{1565664}, W-L Subtree Kernel~\cite{JMLR:v12:shervashidze11a}, RetGK~\cite{arxiv.1809.02670}, and PATCHY-SAN (PSCN)~\cite{pmlr-v48-niepert16}. All methods are 10-fold cross-validated, and we report the average accuracies and run them 10 times with different random splits.

\begin{table*}[t]
\centering
\begin{adjustbox}{width=0.85\textwidth}
\begin{tabular}{@{}|c|c|c|c|c|c|c|@{}}
\hline
Tasks & Labels & Spectral Moments & Shortest Path & W-L Subtree & RetGK & PSCN \\
\hline
8-classes &  8 & \textbf{97.5(0.1)} & 74.2(0.2) & 88.3(0.5) & 95.2(0.2) & 86.7(0.8)\\
coauth & 3 & \textbf{73.4(0.4)} & 54.5(0.2) & 69.1(0.6) & 71.8(0.4) & 55.6(1.0)\\
contact & 2 & \textbf{98.1(0.2)} & 96.9(0.4) & 89.1(1.0) & 97.3(0.4) & 67.7(2.3)\\
email & 2 & \textbf{97.5(0.3)} & 92.8(0.3) & 92.5(1.1) & 93.9(0.3) & 61.7(3.5)\\
tags & 3 & \textbf{78.9(1.2)} & 58.8(0.7) & 45.1(1.5) & 50.1(2.0) & 40.6(3.2)\\
twitter & 2 & \textbf{99.2(0.3)} & 81.5(0.4) & 85.1(1.7) & 89.5(0.8) & 95.7(2.1)\\
\hline
\end{tabular}
\end{adjustbox}
\caption{Accuracy (standard deviation) of Higher-Order Graph Classification\vspace{-8mm}}
\label{tab2:accs_hon}
\end{table*}

As shown in Table~\ref{tab2:accs_hon}, in all tasks, the proposed method \textit{Spectal Moments} on higher-order graphs outperforms other baselines on downgraded graphs. In sum, \textit{higher-order spectrums provide more information on graph structure}.

\vspace{-4mm}\subsection{Assessing the Required Number of Moments}
\vspace{-2mm}We only select three spectral moments of each random walk transition matrix as features. Is that enough? For general classification tasks, there is a trade-off between the number of training attributes and performance. Since spectral moments are defined as the expectation of return probabilities in Equation~\ref{eq2:spectral_moments}, they are not completely independent of each other. Higher-order moments are less informative when the graph is relatively small. For example, if the graph is as simple as a triangle, it is easy to see $m_1=0$ and $m_{n+1}=(1-m_{n})\times 0.5$ for $n=1,2,\dots$. While the first few moments are different, the following moments start to converge to the stationary probability very quickly, which is $1/3$. 

To empirically show whether the number of moments is enough for classification, we increase the number of moments used in each task in Table~\ref{tab2:accs_hon}. Figure~\ref{fig1:accs_via_moments}, plots the accuracy variations by increasing the number of moments. The number of spectral moments ranges from 1 (only $m_2$) to 14 ($m_2$ to $m_{15}$). For tasks \textit{8-classes, coauth, tags} and \textit{twitter}, we reach the best performance using the first few moments and remain stable for more moments. And for tasks \textit{contact} and \textit{email}, the results fluctuate or even drop as we increase the number of moments. Recall that spectral moments indicate the expectation of return probabilities of random walks of length $l$. Under the graph sizes that we sampled (50 to 200), higher-order moments are less likely to help distinguish one graph from others. We conclude that \textit{higher-order graphs can be sufficiently represented as features using the first few (i.e., first 2-4) spectral moments.}

\begin{figure*}[tbp]
\begin{minipage}{.61\textwidth}
\centering
	\resizebox{0.49\linewidth}{!}{% This file was created with tikzplotlib v0.10.1.
\begin{tikzpicture}

\definecolor{darkgray176}{RGB}{176,176,176}
\definecolor{steelblue31119180}{RGB}{31,119,180}

\begin{axis}[
height=4cm,
tick align=outside,
tick pos=left,
width=16cm,
x grid style={darkgray176},
xmin=0.35, xmax=14.65,
xtick style={color=black},
xlabel={Number of Moments},
y grid style={darkgray176},
ylabel={ACCs},
ymin=0.979625263409962, ymax=0.98867023467433,
ticklabel style = {font=\large},
label style={font=\fontsize{17pt}{1em}\color{white!15!black}\selectfont},
ytick style={color=black}
]
\addplot [semithick, steelblue31119180]
table {%
1 0.980036398467433
2 0.986393678160919
3 0.988259099616858
4 0.987637931034482
5 0.987430555555555
6 0.987636973180077
7 0.987430076628352
8 0.987568007662835
9 0.987568007662835
10 0.987430076628353
11 0.987361111111111
12 0.987499042145594
13 0.987015804597701
14 0.98694683908046
};
\end{axis}

\end{tikzpicture}}
	\resizebox{0.49\linewidth}{!}{% This file was created with tikzplotlib v0.10.1.
\begin{tikzpicture}

\definecolor{darkgray176}{RGB}{176,176,176}
\definecolor{steelblue31119180}{RGB}{31,119,180}

\begin{axis}[
height=4cm,
tick align=outside,
tick pos=left,
width=16cm,
x grid style={darkgray176},
xmin=0.35, xmax=14.65,
xtick style={color=black},
xlabel={Number of Moments},
y grid style={darkgray176},
ylabel={ACCs},
ymin=0.70548568729334, ymax=0.735994709494568,
ticklabel style = {font=\large},
label style={font=\fontsize{16pt}{1em}\color{white!15!black}\selectfont},
ytick style={color=black}
]
\addplot [semithick, steelblue31119180]
table {%
1 0.706872461029759
2 0.723281530467643
3 0.727056211620217
4 0.730213037316958
5 0.731382616910723
6 0.732765233821445
7 0.732148795465281
8 0.731251771374587
9 0.732276334435522
10 0.731314596126594
11 0.729538025507794
12 0.730356636750118
13 0.730086915446386
14 0.734607935758148
};
\end{axis}

\end{tikzpicture}}
	\resizebox{0.49\linewidth}{!}{% This file was created with tikzplotlib v0.10.1.
\begin{tikzpicture}

\definecolor{darkgray176}{RGB}{176,176,176}
\definecolor{steelblue31119180}{RGB}{31,119,180}

\begin{axis}[
height=4cm,
tick align=outside,
tick pos=left,
width=16cm,
x grid style={darkgray176},
xmin=0.35, xmax=14.65,
xtick style={color=black},
xlabel={Number of Moments},
y grid style={darkgray176},
ylabel={ACCs},
ymin=0.968121428571429, ymax=0.978021428571429,
ticklabel style = {font=\large},
label style={font=\fontsize{16pt}{1em}\color{white!15!black}\selectfont},
ytick style={color=black}
]
\addplot [semithick, steelblue31119180]
table {%
1 0.977571428571429
2 0.972571428571429
3 0.975428571428572
4 0.974047619047619
5 0.974761904761905
6 0.974761904761905
7 0.975428571428572
8 0.973428571428572
9 0.972666666666667
10 0.973428571428572
11 0.968571428571429
12 0.968571428571429
13 0.969285714285714
14 0.969285714285714
};
\end{axis}

\end{tikzpicture}}
	\resizebox{0.49\linewidth}{!}{% This file was created with tikzplotlib v0.10.1.
\begin{tikzpicture}

\definecolor{darkgray176}{RGB}{176,176,176}
\definecolor{steelblue31119180}{RGB}{31,119,180}

\begin{axis}[
height=4cm,
tick align=outside,
tick pos=left,
width=16cm,
x grid style={darkgray176},
xmin=0.35, xmax=14.65,
xlabel={Number of Moments},
xtick style={color=black},
y grid style={darkgray176},
ylabel={ACCs},
ymin=0.96675, ymax=0.97225,
ticklabel style = {font=\large},
label style={font=\fontsize{16pt}{1em}\color{white!15!black}\selectfont},
ytick style={color=black}
]
\addplot [semithick, steelblue31119180]
table {%
1 0.9705
2 0.972
3 0.967
4 0.967
5 0.9695
6 0.969
7 0.9705
8 0.97
9 0.97
10 0.9695
11 0.97
12 0.97
13 0.97
14 0.9695
};
\end{axis}

\end{tikzpicture}}
	\resizebox{0.49\linewidth}{!}{% This file was created with tikzplotlib v0.10.1.
\begin{tikzpicture}

\definecolor{darkgray176}{RGB}{176,176,176}
\definecolor{steelblue31119180}{RGB}{31,119,180}

\begin{axis}[
height=4cm,
tick align=outside,
tick pos=left,
width=16cm,
x grid style={darkgray176},
xmin=0.35, xmax=14.65,
xtick style={color=black},
xlabel={Number of Moments},
y grid style={darkgray176},
ylabel={ACCs},
ymin=0.598878735632184, ymax=0.803592528735632,
ticklabel style = {font=\large},
label style={font=\fontsize{16pt}{1em}\color{white!15!black}\selectfont},
ytick style={color=black}
]
\addplot [semithick, steelblue31119180]
table {%
1 0.608183908045977
2 0.784919540229885
3 0.786275862068966
4 0.788252873563218
5 0.78864367816092
6 0.789597701149425
7 0.794287356321839
8 0.792632183908046
9 0.790942528735632
10 0.78967816091954
11 0.785275862068965
12 0.783563218390805
13 0.782931034482759
14 0.785632183908046
};
\end{axis}

\end{tikzpicture}}
	\resizebox{0.49\linewidth}{!}{% This file was created with tikzplotlib v0.10.1.
\begin{tikzpicture}

\definecolor{darkgray176}{RGB}{176,176,176}
\definecolor{steelblue31119180}{RGB}{31,119,180}

\begin{axis}[
height=4cm,
tick align=outside,
tick pos=left,
width=16cm,
x grid style={darkgray176},
xlabel={Number of Moments},
xmin=0.35, xmax=14.65,
xtick style={color=black},
y grid style={darkgray176},
ylabel={ACCs},
ymin=0.984968421052631, ymax=0.994347368421053,
ticklabel style = {font=\large},
label style={font=\fontsize{16pt}{1em}\color{white!15!black}\selectfont},
ytick style={color=black}
]
\addplot [semithick, steelblue31119180]
table {%
1 0.985394736842105
2 0.993921052631579
3 0.993921052631579
4 0.993921052631579
5 0.993921052631579
6 0.993921052631579
7 0.993921052631579
8 0.993921052631579
9 0.993921052631579
10 0.993921052631579
11 0.993921052631579
12 0.993921052631579
13 0.993921052631579
14 0.993921052631579
};
\end{axis}

\end{tikzpicture}}
    \captionof{figure}{Accuracy changes with more spectral moments used as features. From top left to bottom right: \textit{8-classes, coauth, contact, email, tags and twitter}. }
    \label{fig1:accs_via_moments}            
\end{minipage}
\quad
\begin{minipage}{.34\textwidth}
    \enspace
    \resizebox{\linewidth}{!}{% This file was created with tikzplotlib v0.10.1.
\begin{tikzpicture}

\definecolor{crimson2143940}{RGB}{214,39,40}
\definecolor{darkgray176}{RGB}{176,176,176}
\definecolor{darkorange25512714}{RGB}{255,127,14}
\definecolor{forestgreen4416044}{RGB}{44,160,44}
\definecolor{lightgray204}{RGB}{204,204,204}
\definecolor{mediumpurple148103189}{RGB}{148,103,189}
\definecolor{steelblue31119180}{RGB}{31,119,180}

\begin{axis}[
height=8cm,
legend cell align={left},
legend style={
  fill opacity=0.8,
  draw opacity=1,
  text opacity=1,
  at={(0.03,0.97)},
  anchor=north west,
  draw=lightgray204
},
tick align=outside,
tick pos=left,
width=11cm,
x grid style={darkgray176},
xlabel={Graph Size},
xmin=-0.2, xmax=4.2,
xtick style={color=black},
xtick={0,1,2,3,4},
xtick={0,1,2,3,4},
xtick={0,1,2,3,4},
xtick={0,1,2,3,4},
xtick={0,1,2,3,4},
xticklabels={5-50,51-100,101-200,201-400,401-800},
xticklabels={5-50,51-100,101-200,201-400,401-800},
xticklabels={5-50,51-100,101-200,201-400,401-800},
xticklabels={5-50,51-100,101-200,201-400,401-800},
xticklabels={5-50,51-100,101-200,201-400,401-800},
y grid style={darkgray176},
ylabel={ACCs},
ymin=0.428995, ymax=0.952705,
ticklabel style = {font=\large},
label style={font=\fontsize{15pt}{1em}\color{white!15!black}\selectfont},
ytick style={color=black}
]
\addplot [semithick, steelblue31119180, mark=*, mark size=3, mark options={solid}]
table {%
0 0.5863
1 0.6959
2 0.7429
3 0.8543
4 0.9289
};
\addlegendentry{Spectral Moments}
\addplot [semithick, darkorange25512714, mark=*, mark size=3, mark options={solid}]
table {%
0 0.5811
1 0.6506
2 0.7194
3 0.7854
4 0.8188
};
\addlegendentry{RetGK}
\addplot [semithick, forestgreen4416044, mark=*, mark size=3, mark options={solid}]
table {%
0 0.5522
1 0.6897
2 0.7351
3 0.8386
4 0.8817
};
\addlegendentry{W-L Subtree}
\addplot [semithick, crimson2143940, mark=*, mark size=3, mark options={solid}]
table {%
0 0.4528
1 0.5483
2 0.5917
3 0.6649
4 0.6567
};
\addlegendentry{Shortest Path}
\addplot [semithick, mediumpurple148103189, mark=*, mark size=3, mark options={solid}]
table {%
0 0.5008
1 0.5314
2 0.5614
3 0.6221
4 0.6432
};
\addlegendentry{PSCN}
\end{axis}

\end{tikzpicture}}
    \captionof{figure}{Accuracies change with various sampled subgraphs sizes. Here, we use \textit{coauth} as an example.}
    \label{fig2:accs_via_size}            
\end{minipage}
\end{figure*}

\vspace{-4mm}\subsection{Preformances with Variations in Graph Size}
\vspace{-3mm}To observe the variations in performance with different graph sizes, we consider the sampled subgraph size ranges: 5-50, 51-100, 101-200, 201-400, and 401-800. We use the "coath" dataset as an example. Figure~\ref{fig2:accs_via_size} illustrates the accuracy of all methods across different graph sizes. We observe that, for all methods, the accuracy increases as the sampled graph size grows larger. This can be attributed to larger graphs containing more information about local structures, which becomes more challenging for fixed-length representations compared to kernel methods.
The kernel methods rely on capturing every single substructure of the entire graph, while our proposed method focuses on extracting the expectation of return probabilities. Despite this difference, our method continues to perform well and exhibits a similar trend to the two kernel methods. On the contrary, the \textit{Shortest Path} and \textit{PSCN} methods experience a performance bottleneck in the size range of 401-800, as depicted in Figure~\ref{fig2:accs_via_size}.

\vspace{-4mm}\section{Conclusion}
\label{sec:discussion}
\vspace{-4mm}We present a comprehensive spectral representation for higher-order networks by leveraging higher-order spectral theory. This representation enables us to compute spectral moments from equivalent matrices of dyadic weighted graphs. We have proved the properties of our proposed features, which provide upper bounds on specific hypergraph properties, including degree and clustering information (number of triangles). These bounds allow us to gain insight into the structural characteristics of higher-order networks. % However, the study of representing the higher-order graph is just at the beginning stage. In the future, we will focus on the following directions: (1) Exploring more detailed meanings of spetral moments along with graph properties; 
%     \item Speedup the algorithms for building the transition matrix of higher-order graphs; 
%     \item Try to extend random walk based graph kernels such as RetGK to higher-order graph classifications.
% \end{itemize}

\vspace{0.5mm}\noindent \textbf{Acknowledgements.} This research was supported in part by the National Science Foundation under award CAREER IIS-1942929.

\vspace{-4mm}\bibliographystyle{splncs04}
\bibliography{reference}

\begin{thebibliography}{10}
\providecommand{\url}[1]{\texttt{#1}}
\providecommand{\urlprefix}{URL }
\providecommand{\doi}[1]{https://doi.org/#1}

\bibitem{DBLP:journals/corr/abs-1802-06916}
Benson, A.R., Abebe, R., Schaub, M.T., Jadbabaie, A., Kleinberg, J.M.: Simplicial closure and higher-order link prediction. CoRR  \textbf{abs/1802.06916} (2018)

\bibitem{DBLP:journals/corr/BensonGL16a}
Benson, A.R., Gleich, D.F., Leskovec, J.: Higher-order organization of complex networks. CoRR  \textbf{abs/1612.08447} (2016)

\bibitem{1565664}
Borgwardt, K., Kriegel, H.: Shortest-path kernels on graphs. In: Fifth IEEE International Conference on Data Mining (ICDM'05). pp. 8 pp.-- (2005)

\bibitem{Chen:2016:XST:2939672.2939785}
Chen, T., Guestrin, C.: {XGBoost}: A scalable tree boosting system. In: KDD. pp. 785--794. KDD '16, ACM, New York, NY, USA (2016)

\bibitem{chung1993laplacian}
Chung, F.: The laplacian of a hypergraph. Expanding graphs pp. 21--36 (1993)

\bibitem{10.1145/3394486.3403060}
Do, M.T., Yoon, S.e., Hooi, B., Shin, K.: Structural patterns and generative models of real-world hypergraphs. p. 176–186. KDD '20, Association for Computing Machinery, New York, NY, USA (2020)

\bibitem{jin2022interpretable}
Jin, S.: Interpretable network representations  (2022)

\bibitem{jin2022spectral}
Jin, S., Ma, R., Li, J., Eftekharnejad, S., Zafarani, R.: A spectral measure for network robustness: Assessment, design, and evolution. In: 2022 IEEE International Conference on Knowledge Graph (ICKG). pp. 97--104. IEEE (2022)

\bibitem{10.1145/3534678.3539433}
Jin, S., Tian, H., Li, J., Zafarani, R.: A spectral representation of networks: The path of subgraphs. p. 698–708. KDD '22 (2022)

\bibitem{inproceedingsJin}
Jin, S., Zafarani, R.: The spectral zoo of networks: Embedding and visualizing networks with spectral moments. pp. 1426--1434. KDD '20 (2020)

\bibitem{DBLP:journals/corr/abs-1903-11835}
Kriege, N.M., Johansson, F.D., Morris, C.: A survey on graph kernels. CoRR  (2019)

\bibitem{leskovec2010signed}
Leskovec, J., Huttenlocher, D., Kleinberg, J.: Signed networks in social media (2010)

\bibitem{lu2011highordered}
Lu, L., Peng, X.: High-ordered random walks and generalized laplacians on hypergraphs (2011)

\bibitem{Milo824}
Milo, R., Shen-Orr, S., Itzkovitz, S., Kashtan, N., Chklovskii, D., Alon, U.: Network motifs: Simple building blocks of complex networks. Science  \textbf{298}(5594) (2002)

\bibitem{DBLP:journals/corr/NarayananCVCLJ17}
Narayanan, A., Chandramohan, M., Venkatesan, R., Chen, L., Liu, Y., Jaiswal, S.: graph2vec: Learning distributed representations of graphs. CoRR  (2017)

\bibitem{pmlr-v48-niepert16}
Niepert, M., Ahmed, M., Kutzkov, K.: Learning convolutional neural networks for graphs. vol.~48, pp. 2014--2023. PMLR, New York, New York, USA (2016)

\bibitem{DBLP:journals/corr/abs-1912-00735}
Pineau, E.: Using laplacian spectrum as graph feature representation. CoRR  (2019)

\bibitem{10.1145/3184558.3186900}
Rossi, R.A., Ahmed, N.K., Koh, E.: Higher-order network representation learning. p. 3–4. WWW '18 (2018)

\bibitem{JMLR:v12:shervashidze11a}
Shervashidze, N., Schweitzer, P., van Leeuwen, E.J., Mehlhorn, K., Borgwardt, K.M.: Weisfeiler-lehman graph kernels. Journal of Machine Learning Research  \textbf{12}(77),  2539--2561 (2011)

\bibitem{10.1145/3682112.3682114}
Tian, H., Zafarani, R.: Higher-order networks representation and learning: A survey. SIGKDD Explor. Newsl.  \textbf{26}(1),  1–18 (Jul 2024)

\bibitem{arxiv.1809.02670}
Zhang, Z., Wang, M., Xiang, Y., Huang, Y., Nehorai, A.: Retgk: Graph kernels based on return probabilities of random walks (2018)

\bibitem{NIPS2006_dff8e9c2}
Zhou, D., Huang, J., Sch\"{o}lkopf, B.: Learning with hypergraphs: Clustering, classification, and embedding. In: Sch\"{o}lkopf, B., Platt, J., Hoffman, T. (eds.) Advances in Neural Information Processing Systems. vol.~19. MIT Press (2006)

\end{thebibliography}
\vspace{-4mm}
\end{document}